\documentclass[letterpaper, reqno,  11pt]{amsart}
 
\usepackage[T1]{fontenc} \usepackage{lmodern,times}
\usepackage{microtype}

\usepackage{comment,url,xcolor, amsmath, amsthm, amssymb,graphicx}

\usepackage{enumitem}
\setlist[enumerate]{itemsep=0.5ex,leftmargin=4ex}

\usepackage[left=1.7in, right=1.7in, top=1.5in, bottom=1.5in]{geometry}

\usepackage{ulem} 
\newcommand{\revised}[1]{{\color{blue}#1}}

\renewcommand{\revised}[1]{#1}

\numberwithin{equation}{section}
\theoremstyle{plain}                
\newtheorem{theorem}{Theorem}[section]
\newtheorem{lemma}[theorem]{Lemma}
\newtheorem{proposition}[theorem]{Proposition}

\theoremstyle{definition}           

\newtheorem{example}[theorem]{Example}

\theoremstyle{remark}
\newtheorem{remark}[theorem]{Remark}

\newcommand{\tot}{\tfrac{1}{2}} 
\newcommand{\oo}[1]{\tfrac{1}{#1}}
\newcommand{\scl}[2]{\langle #1,#2 \rangle} 
\newcommand{\abs}[1]{\left| #1 \right|} 
\newcommand{\ab}[1]{\langle #1 \rangle} 
\newcommand{\set}[1]{\{#1\}} 
\newcommand{\Bset}[1]{\Big\{#1\Big\}} 
\newcommand{\Bsets}[2]{\Bset{#1\,:\,#2}} 

\newcommand{\prf}[1]{ \{ #1 \}_{t\in [0,T]}}

\newcommand{\dd}{d}
\newcommand{\rn}[2]{\frac{\dd #1}{\dd #2}}
\newcommand{\trn}[2]{\tfrac{\dd #1}{\dd #2}}

\newcommand{\downto}{\searrow}
\providecommand{\R}{} \renewcommand{\R}{{\mathbb R}}
\newcommand{\PP}{{\mathbb P}}
\newcommand{\EE}{{\mathbb E}}
\newcommand{\ee}[1]{ \bE \left[ #1 \right] }

\newcommand{\FFF}{{\mathbb F}}
\newcommand{\EN}{{\mathcal E}}
\newcommand{\eps}{\varepsilon}
\newcommand{\ld}{\lambda}
\newcommand{\Ld}{\Lambda}
\newcommand{\el}{{\mathbb L}} 

\newcommand{\lone}{\el^1}

\newcommand{\efor}{\text{ for }}

\newcommand{\eand}{\text{ and }}
\newcommand{\ewhere}{\text{ where }}
\newcommand\bE{{\mathbb E}}
\newcommand\sF{{\mathcal F}}
\newcommand\sH{{\mathcal H}}
\newcommand\tM{{\tilde{M}}}
\newcommand\sP{{\mathcal P}}
\newcommand\sX{{\mathcal X}}
\newcommand\tX{{\tilde{X}}}

\newcommand\sY{{\mathcal Y}}
 
\newcommand{\uppar}[2]{#1^{(#2)}}
\newcommand{\upeps}[1]{\uppar{#1}{\eps}}
\newcommand{\upz}[1]{\uppar{#1}{0}}
\newcommand{\updel}[1]{\uppar{#1}{\delta}}
\newcommand{\Se}{\upeps{S}}
\renewcommand{\Re}{\upeps{R}}
\newcommand{\sXe}{\upeps{\sX}}
\newcommand{\ue}{\upeps{u}}

\newcommand{\Ke}{\upeps{K}}
\newcommand{\ve}{\upeps{v}}
\newcommand{\hXe}{\upeps{\hat{X}}}
\newcommand{\hYe}{\upeps{\hat{Y}}}
\newcommand{\sYe}{\upeps{\sY}}
\newcommand{\lde}{\upeps{\ld}}
\newcommand{\Ze}{\upeps{Z}}
\newcommand{\hHe}{\upeps{\hat{H}}}

\newcommand{\ye}{\upeps{y}}

\newcommand{\Rz}{\upz{R}}

\newcommand{\uz}{\upz{u}}

\newcommand{\Kz}{\upz{K}}
\newcommand{\vz}{\upz{v}}
\newcommand{\hXz}{\upz{\hat{X}}}
\newcommand{\hYz}{\upz{\hat{Y}}}

\newcommand{\Zz}{\upz{Z}}
\newcommand{\hHz}{\upz{\hat{H}}}
\newcommand{\hpiz}{\upz{\hat{\pi}}}
\newcommand{\PPz}{\upz{\tilde{\PP}}}
\newcommand{\Dz}{\upz{\Delta}}

\newcommand{\RRz}{\tilde{\mathbb P}^{(0)}}
\newcommand{\Zd}{\updel{Z}}
\newcommand{\Zb}{\uppar{Z}{\beta}}
\newcommand{\derep}[1]{\trn{}{\eps} #1\Big|_{\eps=0+} }
\newcommand{\sPm}{\sP_{M}}
\newcommand{\tMp}{\tM^p}
\newcommand{\uel}{\ue_{\log}}
\newcommand{\uzl}{\uz_{\log}}
\newcommand{\vel}{\ve_{\log}}
\newcommand{\ce}{\upeps{\mathrm{CE}}}
\newcommand{\ldk}{\ld^\text{KO}}

\begin{document}

\title{An expansion in the model space in the context of utility maximization}

\author{Kasper Larsen}
\address{Kasper Larsen, Department of Mathematical Sciences, Carnegie
Mellon University}
\email{kasperl@andrew.cmu.edu}
\thanks{
  The authors would like to thank Milica \v Cudina, Claus Munk, Mihai S\^\i rbu and Kim Weston for discussions.  During the preparation of
  this work the first author has been supported by the National Science
  Foundation under Grant No.~DMS-1411809 (2014 - 2017).  The second author has been supported by the National Science Foundation under grant No.~DMS-1600307 (2015 - 2018).  The third author has been supported by the NSF under Grants  No.~DMS-0706947 (2010 - 2015) and No.~DMS-1107465 (2012 - 2017).  Any opinions, findings and conclusions or recommendations
  expressed in this material are those of the author(s) and do not
  necessarily reflect the views of the National Science Foundation (NSF)}
 
\author{Oleksii Mostovyi}
\address{Oleksii Mostovyi, Department of Mathematics,  
University of Connecticut}
\email{oleksii.mostovyi@uconn.edu}


\author{Gordan \v{Z}itkovi\'{c}}
\address{Gordan \v Zitkovi\' c, Department of Mathematics,  
University of Texas at Austin}
\email{gordanz@math.utexas.edu}


  
\begin{abstract}  

In the framework of an incomplete financial market where the stock price
dynamics are modeled by a continuous semimartingale (not
necessarily Markovian) an explicit
second-order expansion formula for the power investor's value function -
seen as a function of the underlying market price of risk process - is
provided.  This allows us to provide first-order approximations of the
optimal primal and dual controls.  Two specific calibrated numerical
examples illustrating the accuracy of the method are also given.

\end{abstract}

\maketitle

\section{Introduction} 
\label{sec:intro}

In an incomplete financial setting with noise governed by a continuous
martingale and in which the investor's preferences are modeled by a
negative power utility function, we provide a  second-order Taylor expansion
of the investor's value function with respect to perturbations of the
underlying market price of risk process. We show that tractable models can
be used to approximate highly intractable ones as long as the latter can
be interpreted as perturbations of the former.  As a by-product of our  analysis we explicitly construct first-order approximations of both the primal and the dual optimizers. Finally, we apply our approximation in two numerical
examples.

There are two different ways of looking at our contribution: as a tool to 
approximate the value function and 
perform numerical computations, or as a stability result with applications to statistical estimation. Let us elaborate on these, and the
related work,  in order.

\par{\em An approximation interpretation.} The conditions for existence and uniqueness of the investor's utility optimizers are
well-established (see \cite{KLSX} and  \cite{KraSch99}). However, in
general settings, the numerical computation of the investor's value
function remains a challenging problem. Various existing approaches
include:

\begin{enumerate}

\item In Markovian settings, the value function can typically be
characterized by a HJB-equation.  Its numerical implementation
through a finite-grid approximation is naturally subject to the curse of
dimensionality. Many authors (see \cite{KimOmb96}, \cite{Wac2002}, \cite{ChaVic2005},
\cite{Kra05}, and \cite{Liu2007}) opt for affine and 
quadratic models for
which closed-form solutions exist. Going beyond these specifications
in high-dimensional settings by using PDE-techniques seems to be
very hard computationally.

\item In general (i.e., not necessarily Markovian) complete models,
\cite{CviGouZap2003} and \cite{DetGarRin2003} provide efficient Monte Carlo
simulation techniques based on the martingale method for complete markets
developed in \cite{CoxHuang1} and \cite{KLS}.

\item Other approximation methods are based on various Taylor-type
expansions. The authors of 
\cite{Cam1993} and \cite{CamVic1999} log-linearize the investor's budget constraint as well as the investor's first-order condition for optimality.  \cite{KogUpp2000}  expand
in the investor's risk-aversion coefficient around the log-investor (the myopic investor's problem is known to be tractable even in incomplete settings). 
When solving the HJB-equation numerically (using Longstaff-Schwartz type of techniques)
\cite{BraGoySanStr2005} expand the value function in the wealth variable to
a forth degree Taylor approximation.

\item Based on the duality results in \cite{KLSX}, \cite{HauKogWan2006} provide an upper bound on the error stemming from using sub-optimal strategies. \cite{BicKraMun2013}
propose a method based on minimizing over a subset of dual
elements. This subset is chosen such that the corresponding dual
utility can be computed explicitly and transformed into a feasible primal strategy.

\item \revised{It is also important to mention the recent explosion in
research in asymptotic methods in a variety of different ares in
mathematical finance (transaction costs, pricing, etc.). Since we focus on
model expansion in utility maximization in this paper, we simply point the
reader to some of the most recent papers, namely \cite{AltMuhSon15}, and
\cite{KalMuh15}, and the references therein, for further information.}
\end{enumerate}

In our work, no Markovian assumption is imposed and we deal with
general, possibly incomplete, markets with continuous price
processes. We note that while our results apply only to $p<0$, it is possible to
extend them to $p\in(0,1)$ at the cost of imposing additional integrability
requirements. We do not pursue such an extension; the parameter range
$p\in (0,1)$ which we leave out seems to lie outside the typical
range of risk-aversion parameters observed in practice (see, e.g.,  \cite{Szp86}).
Moreover, we do not consider utility functions more general than the powers. While there are no significant additional mathematical difficulties in treating the general case under appropriate conditions on the relative risk-aversion coefficients, we do not believe that the added value justifies the corresponding 
notational and technical overhead. For example, all our results would become dependent on the agent's initial wealth, and this dependence would permeate the entire analysis.

\par{\em A stability interpretation.} As we mentioned above, our contribution can also be seen as a stability
result.  It is well-known (see, e.g., \cite{Rogers}) that even in
Samuelson's model, estimating the drift is far more challenging than
estimating the volatility. \cite{LarZit07} identify the kinds of
perturbations of the market price of risk process under which the value
function behaves continuously. In the present paper we take the stability
analysis one step further and provide a first-order Taylor expansion in 
an 
infinite-dimensional space of the market price of risk processes. This way,
we not only identify the ``continuous'' directions, but also identify
those features of the market price of risk process that affect the
solution of the utility maximization problem the most (at least locally).  
Any statistical procedure which is performed with utility maximization in
mind should, therefore, focus on those, salient, features in order to
use the scarce data most efficiently. 

Similar perturbations have been considered by \cite{Monoyios},
but in a somewhat different setting. \cite{Monoyios} is based on Malliavin
calculus and produces a first-order expansion for the
utility-indifference price of an exponential investor in an It\^o-process
driven market; some of the ideas used can be traced to the related
work  \cite{Dav06}.

\par{\em Mathematical challenges.} From a mathematical point of view,
our approach is founded on two ideas.  One of them is to extend the
techniques and results of \cite{LarZit07}; indeed, the basic fact that the
optimal dual minimizers converge when the market prices of risk process does is
heavily exploited. It does not, however, suffice to get the full picture.
For that, one needs to work on the primal and the dual problems
simultaneously and use a pair of bounds. The ideas used there are related
to and can be interpreted as a nonlinear version of the primal-dual
second-order error estimation techniques first used in \cite{Hen2002} in
the context of mathematical finance.  The first-order expansion in the
quantity of the unspanned contingent claim developed in \cite{Hen2002} was 
generalized in 
\cite{KramkovSirbu2006b}
(see also \cite{KramkovSirbu2006a}). The arguments in these papers rely on convexity and
concavity properties in the expansion parameter (wealth and number of
unspanned claims). This is not the case in the present paper; indeed, when seen as a function of the underlying
market price of risk process, the investor's value function is neither
convex nor concave and a more delicate, local, analysis needs to be
performed.

\bigskip

\par{\em Numerical examples.} 
In Section \ref{sec:examples} we use two examples to
illustrate how our approximation performs under realistic conditions.
First, we consider the  Kim-Omberg model (see \cite{KimOmb96}) which is widely
used in the financial literature. Under a calibrated set of parameters, 
we find that our approximation is indeed very accurate when compared to the exact values. 

Our second example belongs to a class of extended
affine models introduced in \cite{CheFilKim07}. The authors  
show that this class of models has superior empirical properties when 
compared to popular affine and quadratic specifications 
(such as those used, e.g.,  in \cite{Liu2007}). 
The resulting optimal investment problem for the extended affine models,
unfortunately, does not seem to be explicitly solvable. Our approximation
technique turns out to be easily applicable and our error bounds are quite
tight in the relevant parameter ranges.

\section{A family of utility-maximization problems}
\label{sec:problem}

\subsection{The setup}
\label{sse:setup}
We work on a filtered probability space $(\Omega,\sF,
\FFF= \prf{\sF_t}, \PP)$, 
with the
finite time horizon $T>0$. We assume that the filtration $\FFF$ is
right-continuous and that the $\sigma$-algebra $\sF_0$ consists of all
$\PP$-trivial subsets of $\sF$. 

Let $M$ be a continuous local martingale, and let $\Re$, $\eps\geq
0$
be a family of continuous $\FFF$-semimartingales given by
\begin{equation}
\label{equ:7283}
\begin{split}
 \Re := M + \int_0^{\cdot} \lde_t\, d\ab{M}_t, \text{ on } 
 [0,T],\ewhere \lde: = \ld  + \eps \ld',
\end{split}
\end{equation}
for a pair $\ld,\ld'\in\sPm^2$, where $\sPm^2$ denotes the collection of 
  all progressively
measurable processes $\pi$ with $\int_0^T \pi_t^2\,
d\ab{M}_t<\infty$.
As $\Se := \EN(\Re)$ (where $\EN$ denotes the stochastic exponential)
will be interpreted as the price process of a
financial asset, the assumption that $\lde\in\sPm^2$ can be taken 
as a minimal no-arbitrage-type condition.
We remark right away that further integrability conditions on $\ld$ and $\ld'$ will need to be
imposed below for our main results to hold.
\subsection{The utility-maximization problem}
Given $x>0$ and $\eps \in [0,\infty)$, let $\sXe(x)$ denote the set of all
nonnegative 
wealth processes starting from initial wealth $x$ in the financial market consisting of 
$\Se := \EN(\Re)$ and a zero-interest bond, i.e., 
\[ \sXe(x) := \Bsets{x \EN\big(\textstyle \int_0^T \pi_t\, d\Re_t\big)}{  \pi \in\ \sPm^2}.\]
Here, $\pi$ is interpreted as the fraction of wealth invested in the risky asset $\Se$.
The investor's preferences are modeled by a CRRA (power) utility function
with the risk-aversion parameter $p<0$:
\begin{equation}\label{equ:U}
U(x) := \frac{x^p}{p}, \quad x>0.
\end{equation}
The value function of the corresponding optimal-investment problem is defined
by
\begin{equation}\label{equ:ue}
\ue(x) := \sup_{X\in\sXe(x)} \EE[ U(X_T)],\ x>0.
\end{equation}
\subsection{The dual utility-maximization problem}
As is usual in the utility-ma\-xi\-mi\-za\-ti\-on literature, a fuller
picture is obtained if one also considers the appropriate version of the
optimization problem dual to \eqref{equ:ue}. For that, we need to examine
the no-arbitrage properties of the set of models introduced in Section
\ref{sse:setup} above.

We observe, first, that the assumptions we placed on the
market price of risk processes $\lde$ above are  not sufficient to
guarantee the existence of an equivalent martingale measure (NFLVR). They
do preclude so-called ``arbitrages of the first kind'' and imply the
related condition NUBPR.  In particular, for all $x,y>0$ and $\eps\geq 0$
there exists a (strictly) positive c\` adl\` ag supermartingale $Y$ with
the property that $Y_0=y$ and $YX$ is a supermartingale for each $X\in\sXe(x)$; we denote
the set of all such processes by $\sYe(y)$. While this is a consequence of the
condition NUBPR in general, in this case an example of a process in
$\sYe(y)$ is given, explicitly, as $y\Ze$, where $\Ze$ is the \emph{minimal} local martingale density:
 \begin{equation}
 \label{equ:Ze}
 \begin{split}
   \Ze = \EN(-\int_0^{\cdot} \lde_t\, dM_t). 
 \end{split}
 \end{equation}

Having described the dual domain, we remind the reader that the
\emph{conjugate} utility function $V:(0,\infty)\to\R$ is defined by
\begin{equation}\label{equ:V}
V(y) := \sup_{x>0} \left( U(x) - xy\right) = \frac{y^{-q}}{q},\ewhere
q:=\tfrac{p}{1-p}\in(-1,0).
\end{equation}
We define the \emph{dual value function} $\ve:(0,\infty)\to \R$ by
\begin{equation}
\label{equ:ve}
\begin{split}
 \ve(y) := \inf_{Y\in \sYe(y) } \EE[ V(Y_T)],\ y>0,\ \eps\geq 0.
\end{split}
\end{equation}
Due to negativity (and, a fortiori, finiteness) of the primal value
function $\ue$, the (abstract) Theorem 3.1 of \cite{KraSch99} can now be applied
(see also \cite{Mostovyi2011}). Its 
main assumption, namely the
bipolar relationship between the primal and dual domains, holds due to the
existence of the num\' eraire process, given explicitly by $1/\Ze$ (see
Theorem 4.12~in \cite{KarKar07}). One can also use 
a simpler argument (see \cite{Lar11}), which applies only to 
the case of a CRRA utility with $p<0$, to obtain the following conclusions
for all $\eps\geq 0$:
\begin{enumerate}
  \item  
both $\ue$ and $\ve$ are
finite 
and the following conjugacy
relationships hold 
 \begin{equation}
 \label{equ:conj}
 \begin{split}
   \ve(y) = \sup_{x>0} \Big( \ue(x)-xy \Big),\eand
   \ue(x) = \inf_{y>0} \Big( \ve(y)+xy \Big).
 \end{split}
 \end{equation}
\item For all $x,y>0$
there exist optimal solutions $\hXe(x)\in\sXe(x)$
and $\hYe(y)\in\sYe(y)$ of \eqref{equ:ue} and \eqref{equ:ve},
respectively,  and are related by
\[ U'(\hXe_T(x)) = \hYe_T (\ye(x))  \ewhere \ye(x) = \trn{}{x} \ue(x) = p
x^{p-1} \ue(1). \]
\item \revised{The product $\hXe\hYe$ is a uniformly-integrable martingale. 
In particular \[ \EE[ \hXe_T \hYe_T]=xy.\]} 
\end{enumerate}
The homogeneity of the utility function $U$ and its conjugate
$V$ transfers to the value functions $\ue$ and $\ve$ and the optimal
solutions $\hXe$ and $\hYe$:
 \begin{equation}
 \label{equ:simpl}
 \begin{split}
    \ue(x) &= x^p \ue,\quad \ve(y)=y^{-q} \ve,\\
    \hXe(x)&=x \hXe,\quad \hYe(y)=y \hYe,
 \end{split}
 \end{equation}
where, to simplify the notation, we write $\ue, \ve, \hXe$ and $\hYe$ for
$\ue(1)$,
$\ve(1)$, $\hXe(1)$ and $\hYe(1)$, respectively. 
\subsection{A change of measure}
For $\eps=0$ we denote by $\hpiz$ the primal optimizer, i.e., the process
in $\sPm^2$ such that
\[ \hXz = \EN( \int_0^{\cdot} \hpiz_u\, d\Rz_u).\]
We define the probability measure $\PPz$ 
 by
 \begin{equation}
 \label{equ:Girsanov}
 \begin{split}
   \rn{\PPz}{\PP} = \hXz_T \hYz_T\, \Big(= \oo{\vz} V(\hYz_T) = \oo{\uz}
U(\hXz_T)\Big),
 \end{split}
 \end{equation}
 \revised{where the last two equalities 
 follow from the  identities $x U'(x) = p U(x)$ and $ y V'(y) = -q V(y)$, and the relations between the value functions outlined above.}

The measure $\PPz$ has been in the mathematical finance literature for a while (see, e.g., p. 911-2 in \cite{KraSch99}).
The explicit form of $\PPz$ is not generally available, but, we note that, by Girsanov's Theorem (see \eqref{equ:ZENL} and the discussion around it),
the process
\begin{equation}
 \label{equ:tMp}
 \begin{split}
   \tMp := M + \int_0^{\cdot} \Big( \ld_t - \hpiz_t \Big)\, d\ab{M}_t 
 \end{split}
 \end{equation}
is a $\PPz$-local martingale; \revised{this fact will be used below in the proof of Proposition \ref{pro:u-down}.} 

\section{The problem and the main results}

We first provide first-order expansions and error estimates of the primal and dual value functions. Secondly, we provide an expansion of the optimal controls in the Brownian setting.

\subsection{Value functions}
At the basic level, we are interested in the first-order properties
of the convergence, as $\eps\downto 0$, of the value
functions of the problems $\ue$ and $\ve$ to the value functions $\uz$ and
$\vz$ of the
``base'' model (corresponding to $\eps=0$). To familiarize ourselves with
the flavor of the results we can expect in the general case, we start by
analyzing a similar problem for the logarithmic utility. It has the
advantage that it admits a simple explicit solution. Let $\uel(x)$ and
$\vel(y)$
denote the value function of the utility maximization problem as in 
\eqref{equ:ue} and \eqref{equ:ve} above, but with $U(x)=\log(x)$ and
$V(y) = \sup_x (U(x)-xy)=-\log(y) -1$.
It is a  classical result that, as long as $\EE[ \int_0^T (\ld_t^2 + (\ld'_t)^2)\, d\ab{M}_t]<\infty$,  we 
have
\[ \uel(x) = \log(x)+\tot \EE[ \int_0^T (\lde_t)^2\, d\ab{M}_t] \eand
\vel = \uel -1.\] The (exact) second-order expansion in $\eps$ 
of $\uel(x)$ is thus
given by
 \begin{equation*}
 \label{equ:4AA5}
 \begin{split}
    \uel(x) &= \uzl(x) + \eps \EE[ \int_0^T \ld_t \ld'_t\, d\ab{M}_t] + \tot \eps^2
   \EE[ \int_0^T (\ld'_t)^2\, d\ab{M}_t]\\
    &=\uzl(x) + \eps \EE[ \int_0^T \ld'_t\, d\Rz_t] + \tot \eps^2
   \EE[ \int_0^T (\ld'_t)^2\, d\ab{M}_t],\\
 \end{split}
 \end{equation*}
 where $\Rz$ is defined in (\ref{equ:7283}).
 We cannot expect the value function to be a second order polynomial in
 $\eps$ in the case of a general power utility. We do obtain a formally similar
 first-order expansion in Theorem \ref{thm:main1} below and an analogous
 error estimate in Theorem \ref{thm:main2}. Section 5 
  is devoted to their proofs. We remind the reader of the homogeneity
 relationships in \eqref{equ:simpl}; they allow us to assume from now on
 that $x=y=1$.
\begin{theorem}[The G\^ateaux derivative]
\label{thm:main1}
In the setting of Section \ref{sec:problem}, we assume that 
\begin{equation}
\label{equ:A1}
\begin{split}
\int_0^T(\ld'_t)^2\, d\ab{M}_t \in \el^{1-p}(\PP) \eand \int_0^T \ld'_t\, d\Rz_t
\in \cup_{s>(1-p)} \el^s(\PP).
\end{split}
\end{equation}
Then, with
$\Dz := \EE^{\PPz}[\textstyle \int_0^T \ld'_t\, d\Rz_t]$, where $\PPz$ is defined by \eqref{equ:Girsanov}, we have
\begin{align}
 \label{equ:main1-u}
 \derep{\ue}:=\lim_{\eps \downto 0}  \oo{\eps} \Big( \ue - \uz \Big) 
 &= p \uz \Dz, \eand \\
\label{equ:main1-v}
 \derep{\ve}:=\lim_{\eps \downto 0}  \oo{\eps} \Big( \ve - \vz \Big) 
 &= q\vz \Dz.\qedhere
\end{align}
\end{theorem}
\begin{theorem}[An error estimate]
\label{thm:main2} 
In the setting of Section \ref{sec:problem}, we assume that 
\begin{equation}
\label{equ:A2}
\begin{split}
\int_0^T(\ld'_t)^2\, d\ab{M}_t, \int_0^T \ld'_t\, d\Rz_t
\in \el^{2(1-p)}(\PP) \eand \Phi^2 e^{\eps_0 |p| \Phi^-}\in\lone(\PPz),
\end{split}
\end{equation}
for some $\eps_0>0$, where $\Phi := \int_0^T \hpiz_t \ld'_t\, d\ab{M}_t$. 
Then there exist constants $C>0$ and $\eps_0'\in (0,\eps_0]$ such that for all $\eps\in [0,\eps'_0]$ we have
\begin{align}
\label{equ:main2-u}
\Big| \ue - \uz - \eps p  \uz \Dz \Big| & \leq 
 C \eps^2, \eand
 \\ \label{equ:main2-v}
\Big| 
\ve - \vz - \eps q \vz \Dz \Big| &
 \leq C  \eps^2.
 \end{align}
\end{theorem}
\begin{remark}\ 
\label{rem:761C}
\begin{enumerate}
\item
It is perhaps more informative to think of the results in Theorems
\ref{thm:main1} and \ref{thm:main2} on the logarithmic scale. 
As is evident from \eqref{equ:main1-u} and
\eqref{equ:main1-v}, the functions $\ue$ and $\ve$ admit the right
\emph{logarithmic} derivative $p\Dz$ and $q\Dz$, respectively,  at
$\eps=0$. Moreover, we have the following small-$\eps$ asymptotics:
\[ 
\ue = \uz e^{\eps p \Dz + O(\eps^2)} \eand
\ve = \vz e^{\eps q \Dz + O(\eps^2)}.\]
If one takes one step further and uses the \emph{certainty
equivalent} $\ce$, given by
\[ U(\ce) = \ue,\]
we note that $\Dz$ is precisely the infinitesimal growth-rate of $\ce$ at
$\eps=0$ - an $\eps$-change of the market price of risk in the direction
$\ld'$ yields to an $e^{\eps \Dz}$-fold increase in the certainty-equivalent
of the initial wealth. 
\item \label{rem:main2}
A careful analysis of 
the proof of Theorem \ref{thm:main2} below reveals the following,
additional, information:
\begin{enumerate}
  \item The proof of Proposition \ref{pro:u-down} reveals that
  $\Dz=\EE^{\PPz}[\Phi]$. 
  \item The condition involving $\Phi$ in \eqref{equ:A2} is needed only for
  the upper bound in \eqref{equ:main2-u} and the lower bound in
  \eqref{equ:main2-v}. The other two bounds hold for all 
  $\eps\geq 0$ even 
  if \eqref{equ:A2} holds with $\eps_0=0$. 
  \item The constants $C$ and $\eps_0'$ depend - in a simple way - on
  $\eps_0$, $p$ and the $\el^{2(1-p)}(\PPz)$- and 
  $\lone(\PPz)$-bounds of the random variables in \eqref{equ:A2}. 
  For two one-sided bounds, explicit formulas are given in Propositions \ref{pro:v-up}
  and \ref{pro:u-down}. The other two bounds are somewhat less informative so we do not compute them explicitly. The reader will find an example of how this can be done in a specific setting in Subsection \ref{sse:extaff}. 
   \item Even though we cannot claim that the functions $\ue$ and $\ve$ are convex or concave, it is possible to show their local {\it semiconcavity} in $\eps$  (see \cite{Can04}).
   This can be done via the techniques from the proof of  Theorem~\ref{thm:main2}.
\end{enumerate} 
\item
The assumption of constant risk aversion (power utility) 
allows us to incorporate many stochastic interest-rate models into our setting.
Indeed, provided that $c:= \ee{e^{p\int_0^T
r_tdt}} <\infty$, we can introduce the probability measure $\PP^r$, defined by
\begin{align}
\rn{\PP^r}{\PP}:= c e^{p\int_0^T r_t\, dt},
\end{align}
on $\sF_T$. For any admissible wealth process $X$ we then have
\[ \EE[ U(X_T) ] = c\, \EE^{\PP^r}\left[ U\Big( X_T e^{-\int_0^T r_u\, du}\Big) \right].\]
This way, the utility maximization under $\PP^r$ with a zero interest rate
becomes equivalent to the utility maximization problem under $\PP$ with
the interest rate process $\prf{r_t}$. \cite{Zitkovic2005} and \cite{Mostovyi2011} consider the setting of utility maximization  with stochastic utility which embeds stochastic interest rates. 

Practical implementation of the above idea depends on how explicit one can
be about the Girsanov transformation associated with $\PP^r$. It turns out,
fortunately, that many of the widely-used interest-rate models, 
such as Vasi\v cek, CIR, or the quadratic normal models (see, e.g.,
\cite{Mun13} for a textbook discussion of these models)
allow for a fully
explicit description (often due to their affine structure). For example, in
the Vasi\v cek model, the Girsanov drift under $\PP^r$ can be computed
quite explicitly, due to the underlying affine structure.
 Indeed, suppose that $r$ has the Ornstein-Uhlenbeck dynamics of
the form
\[ dr_t := \kappa (\theta-r_t)\, dt + \beta\, dB_t,\ r_0\in\R,\]
where $B$ is a Brownian motion and $\kappa>0$, $\theta,\beta\in\R$.
Then the process
\[ B^{(p)} := B - \int_0^{\cdot} b(T-t)\, dt,\ewhere
b(t) = \tfrac{\beta p}{\kappa} (1-e^{-\kappa t}),\]
is a $\PP^r$-Brownian motion. 
\end{enumerate}
\end{remark}

\subsection{Optimal controls}
The estimates \eqref{equ:main2-u} and \eqref{equ:main2-v} are of type $O(\eps^2)$. A slight adjustment to the below proof of Proposition \ref{pro:u-down}  shows that the wealth process $\tilde{X} :=\mathcal{E}\big(\int\hat{\pi}^{(0)}dR^{(\eps)})$ satisfies (see \ref{lower_bound_est})
\begin{align*}
\big|\EE[U(\tilde{X}_T)] - u^{(0)}(1+\eps p \Delta^{(0)}) \big|&\le \frac12p^2\eps^2 |u^{(0)}| \EE^{\PPz}[\Phi^2e^{\eps |p| \Phi^-}].
\end{align*}

Therefore, under the conditions of Theorem \ref{thm:main2}, $\hat{\pi}^{(0)}$ is an $O(\eps^2)$-optimal control for the $\eps$-model because the triangle inequality produces a constant $C>0$ such that
$$
\big|\EE[U(\tilde{X}_T)] -u^{(\eps)} \big| \le C \eps^2,
$$
for all $\eps>0$ small enough. In this section we will provide a correction term to $\hat{\pi}^{(0)}$ such that the resulting wealth process upgrades the convergence  to $o(\eps^2)$. 

For simplicity, we consider the (augmented) filtration generated by $(B,W)$ where $B\in \R$ and $W\in\R^d$, $d\in\mathbb{N}$, are two independent Brownian motions. In \eqref{equ:7283} we take 
\begin{align}\label{dM_BM}
dM_t := \sigma_t dB_t,\quad M_0:=0,
\end{align}
for a process $\sigma\in \sP_{B}^2$ with $\sigma\neq0$. We define $\PPz$ by
  \eqref{equ:Girsanov} and we denote by $(B^{\PPz},W^{\PPz})$ the
  corresponding $\PPz$-Brownian motions. Provided that $\Phi :=\int_0^T
  \hpiz_t \ld'_t\sigma^2_t dt \in \el^{2}(\PPz)$, $\Phi$ has the unique
  martingale representation under $\PPz$
\begin{align}\label{Itorep}
\Phi = \EE^{\PPz}[\Phi] + \int_0^T \gamma^B_{t}\sigma_tdB^{\PPz}_t  +  \int_0^T \gamma^W_{t} dW^{\PPz}_t,
\end{align}
where we have used $\sigma\neq0$. Because $\Phi\in \el^{2}(\PPz)$ the two processes $\gamma^W$ and $\gamma^B$ in \eqref{Itorep} satisfy the integrability conditions
$$
\EE^{\PPz}\Big[\int_0^T\Big( (\gamma^B_{t}\sigma_t)^2 + (\gamma^W_t)^2\Big)dt\Big]<\infty.
$$
These square integrability properties will be used in the proof of the next theorem.

\begin{theorem}[2nd order expansion] \label{thm:main3}  In the above Brownian setting, we assume
\begin{align}\label{2nd_int}
\int_0^T(\ld'_t)^2\sigma^2_tdt \in \el^{1-p}(\PP)\cap \el^{1}(\PPz) 
\eand \int_0^T  \hat{\pi}_t^{(0)} \lambda_t' \sigma_t^2 dt \in \el^{2}(\PPz),
\end{align}
as well as the existence of a constant $\eps_0 >0$ such that $\delta :=\tfrac{\lambda'+p\gamma^B}{1-p}$ satisfies
\begin{align}\label{exp_reg}
e^{p\int_0^T\big (\eps\hat{\pi}^{(0)} \lambda'+ \eps^2 (\delta \lambda'-\frac12\delta^2 ) \big)\sigma^2dt + p\eps \int_0^T\delta\sigma dB_t^{\PPz} }\in \el^{1}(\PPz),
\end{align}
for all $\eps \in (0,\eps_0)$. Then we have
\begin{align}
\label{equ:main3-u}
 \ue - \uz - \eps p  \uz \Dz - \tfrac12 \eps^2p\uz\left(\Delta^{(00)}+p(\Delta^{(0)})^2 \right)& \in O(\eps^3),
 \\ \label{equ:main3-v}
\ve - \vz - \eps q \vz \Dz  - \tfrac12 \eps^2q\vz\left(\Delta^{(00)}+q(\Delta^{(0)})^2 \right) &
\in O(\eps^3),
 \end{align}
as $\eps\downto 0$. In \eqref{equ:main3-u} and \eqref{equ:main3-v} we have defined
 \begin{align}
\Delta^{(00)} &:= \EE^{\PPz}\left[ \int_0^T\left(p|\gamma^W_{t}|^2 +\frac{(\lambda'_t)^2 + p\gamma^B_t(\gamma^B_t+2\lambda'_t)}{1-p}\sigma^2_t\right)dt\right],\label{Delta00}
\end{align}
where the processes $\gamma^B$ and $\gamma^W$ are given by the
martingale representation \eqref{Itorep}.
\end{theorem}

\begin{remark}\ 
\label{rem:controls}
\begin{enumerate}
\item The below proof of Theorem \ref{thm:main3} shows that the process
\begin{align}\label{optimal_control_correction}
\tilde{\pi} := \hat{\pi}^{(0)} + \eps \tfrac{\lambda'+p\gamma^B}{1-p},
\end{align}
is an $O(\eps^3)$-optimal control for the $\eps$-model in the sense that the wealth process $\tilde{X} :=\mathcal{E}\big(\int\tilde{\pi}dR^{(\eps)})$ satisfies
$$
\EE[U(\tilde{X}_T)] -u^{(\eps)} \in O(\eps^3)\quad \text{as}\quad \eps \downto 0.
$$
\item Because the filtration is generated by $(B,W)$, the optimizer $\hat{H}^{(0)}$  for the dual problem \eqref{equ:ve} can be written as $\hat{H}^{(0)} = \mathcal{E}(-\int \hat{\nu}^{(0)}dW)$ for a $d$-dimensional process  $\hat{\nu}^{(0)}$ in $\sP_{W}^2$. The below proof of Theorem \ref{thm:main3} also shows that the process
\begin{align}\label{optimal_dual_control_correction}
\tilde{\nu} := \hat{\nu}^{(0)} - \eps p\gamma^W,
\end{align}
is an $O(\eps^3)$-optimal dual control in the $\eps$-model.

\item Throughout the paper we have considered $\eps=0$ as the base model. Because we can write 
$$
\lambda + (\bar{\eps} +\eps)\lambda' = \lambda + \bar{\eps}\lambda' +\eps\lambda',
$$
for any $\bar{\eps}\in[\eps_L,\eps_U]$ with $\eps_L<\eps_U$, we can use Theorem \ref{thm:main3} for the base model $\lambda+ \bar{\eps}\lambda'$ to  provide a 2nd order Taylor expansion around any point $\bar{\eps}$. Therefore, whenever $\Delta^{(0)}$ and $\Delta^{(00)}$ are bounded uniformly in $\bar{\eps}\in[\eps_L,\eps_U]$, Theorem 3 in \cite{Oli54} ensures that $u^{(\eps)}$ is twice differentiable in $\eps$.
\item An easy way of eliminating the stochastic $B^{\PPz}$-integral  
appearing in \eqref{exp_reg} is to use H\"older's inequality with the exponents $-1/q$ and $(1-p)$; see Section \ref{sse:extaff} below for an example.
\end{enumerate}
\end{remark}

\section{Proofs of the main theorems}
\label{sec:proofs}
We start the proof with a short discussion of the special structure 
the dual domain $\sYe$ has when the stock-price process $\Se = \mathcal{E}(\Re)$ is
continuous.  Indeed, it has been shown in \cite{LarZit07}, Proposition 3.2, p.~1653,  that in that case the maximal
elements in $\sYe$ (in the pointwise order) are precisely 
local martingales of the form
 \begin{equation*}
 \begin{split}
   Y = \Ze H ,\ H\in\sH,
 \end{split}
 \end{equation*}
 where $\sH$ denotes the set of all $M$-orthogonal positive local
 martingales $H$ with $H_0=1$. 
We remark that even though the results in
\cite{LarZit07} were written under the assumption NFLVR, a simple 
localization argument shows that they apply under the present conditions,
as well.
   Hence, we can write
  \[ \ve = \inf_{H\in\sH} \EE[ V(\Ze_T H_T)],\]
and the minimizer $\hYe$ always has the
form \begin{equation}
 \label{equ:ZENL}
 \begin{split}
   \hYe = \Ze \hHe,\text{
for some $\hHe\in\sH$.}
 \end{split}
 \end{equation}
Finally, we introduce two shortcuts for expressions that appear
frequently in the proof:
\begin{equation}
\label{equ:etaLambda}
\begin{split}
 \eta := \int_0^T \ld'_t\, d\Rz_t,\ \Lambda := \int_0^T (\ld'_t)^2\,
 d\ab{M}_t,
\end{split}
\end{equation}
and remind the reader that
$\Phi := \int_0^T \hpiz_t \ld'_t\, d\ab{M}_t$ and 
 $\Delta^{(0)} := \EE^{\PPz}[ \eta]$. It will be useful
to keep in mind that $(1-p)(1+q)=1$ and that $-1/q$ and $1-p$ are
conjugate exponents. 
\subsection{A proof of Theorem \ref{thm:main1}}
The proof is based on the stability results of \cite{LarZit07} and the
following lemma:
\begin{lemma}
\label{lem:K-eps}
Let $\set{\Ke}_{\eps\geq 0}$ be a family of 
positive random variables such that 
\begin{enumerate}
  \item $\EE[ \Zd_T \Ke]\leq 1$ for all $\eps,\delta \geq 0$, and 
  \item $\Ke  \to \Kz$ in probability, as $\eps \downto 0$. 
\end{enumerate}
Then, under the conditions of Theorem \ref{thm:main1}, we have
\[ \lim_{\eps\downto 0} \oo{\eps}
\EE\Big[ V(\Ze_T \Ke) - V(\Zz_T \Ke)\Big] = 
q \ee{ V(\Zz_T \Kz)\eta}.\]
 \end{lemma} 
 \begin{proof}
The map
$\eps\mapsto \Ze_T$ is almost surely
continuously differentiable; \revised{
indeed, we have 
\[
\log(\Ze_T) 
= \log(\Zz_T)
  - \eps  \int_0^T \ld'_t\, d\Rz_t 
  - \tot \eps^2 \int_0^T (\ld'_t)^2\, d\ab{M}_t,
\]
and, so, 
\[
\tfrac{d}{d \eps} \Ze_T = -\Ze_T \Big( \eta + \eps \Lambda \Big), \text{ a.s.}
\]
Therefore, }
\begin{align}
\label{equ:primal1}
V(\Ze_T K) -V(\Zz_T K)
 = \int_0^\eps q V(\Zd_T K) (\eta +\delta \Lambda)d\delta,
\end{align} 
for each $\eps$ and each positive random variable $K$. 
Thus, 
\begin{equation}
\label{equ:378}
\begin{split}
V(\Ze_T \Ke) & -V(\Zz_T \Ke) - \eps q V(\Zz_T \Kz)\eta 
=
A_{\eps}+B_{\eps},
\end{split}
\end{equation}
where
\begin{equation}
\label{equ:A-B}
\begin{split}
  A_{\eps}&:=\int_0^\eps
  q \Big( V(\Zd_T \Ke) - V(\Zz_T \Kz) \Big)\eta\, d\delta, \eand \\
  B_{\eps}&:= \int_0^{\eps}q V(\Zd_T\Ke) \Lambda\,  \delta\, d\delta. 
\end{split}
\end{equation}
H\" older's inequality implies that
\begin{equation}
\label{equ:Holder}
\begin{split}
  \EE[ B_{\eps}] \leq \tot  \eps^2 \sup_{\delta\in [0,\eps]} \Big(\EE[ \Zd_T \Ke]^{-q} \EE[
  \Lambda^{1-p}]^{1+q}\Big)
  \leq \tot \eps^2 \EE[ \Ld^{1-p}]^{1+q}.
\end{split}
\end{equation}
Thus, 
we have
$\tfrac{1}{\eps} \EE[ B_{\eps}] \to 0$, as $\eps\downto 0$. 
To show that  $\tfrac{1}{\eps} \EE[ A_{\eps}]\to 0$, we note that 
$\EE[ A_{\eps}] = \int_0^{\eps} f(\eps,\delta)\, d\delta$, where
the function
$f:[0,\infty)^2 \to \R$ is given by 
 \begin{equation}
 \label{equ:function-f}
 \begin{split}
   f(\eps,\delta) := q\ee{ \big(V(\Zd_T\Ke) - V(\Zz_T \Kz) \big)\eta}.
 \end{split}
 \end{equation}
Since $f(0,0)=0$, it will be enough to show that $f$ is continuous at $(0,0)$.  
By the assumptions of the lemma and the definition of $\Zd$, we have
\[ V(\uppar{Z}{\delta_n}_T \uppar{K}{\eps_n}) \to V(\Zz_T \Kz),\text{ in
probability},\]
for each sequence $(\eps_n,\delta_n)\in [0,\infty)^2$ such
  $(\eps_n,\delta_n)\to (0,0)$.
  Therefore, it suffices to
establish uniform integrability of the expression
inside of the expectation in \eqref{equ:function-f}. For that we
can use the theorem of de la Valle\' e-Poussin, 
 whose conditions hold
thanks to an application H\"older's inequality as in \eqref{equ:Holder}
above, remembering that not only $\eta\in\el^{1-p}$, but
also in $\el^s$, for some $s>(1-p)$. 
\end{proof}
\begin{proof}[Proof of Theorem \ref{thm:main1}] 
  Thanks to the optimality of $\Ze_T \hHe_T$, we have the upper estimate
   \begin{equation}
   \label{equ:upper}
   \begin{split}
     \oo{\eps}\ee{ V(\Ze_T \hHe_T) - V(\Zz_T \hHz_T)} &\leq
     \oo{\eps}\ee{ V(\Ze_T \hHz_T) - V(\Zz_T \hHz_T)}
   \end{split}
   \end{equation}
   Similarly, we obtain the lower estimate
   \begin{equation}
   \label{equ:lower}
   \begin{split}
     \oo{\eps}\ee{ V(\Ze_T \hHe_T) - V(\Zz_T \hHz_T)} &\geq
     \oo{\eps}\ee{ V(\Ze_T \hHe_T) - V(\Zz_T \hHe_T)}.
   \end{split}
   \end{equation}
Our next task is to prove that the limits of the right-hand sides of 
\eqref{equ:upper} and \eqref{equ:lower} exist and both coincide with the 
right-hand side of \eqref{equ:main1-v}. In each case, Lemma \ref{lem:K-eps}
can be applied; in the first with $\Ke= \hHz_T$, and in the second with
$\Ke=\hHe_T$. In both cases the assumption (1) of Lemma
\ref{lem:K-eps} follows directly from that fact that $\Ze_T \Ke \in \sYe$.
As for the assumption (2), it trivially holds in
the first case. In the second case, we need to argue that
$\hHe_T \to \hHz_T$ in probability, as $\eps\downto 0$. That, in turn,
follows easily from Lemma 3.10 in \cite{LarZit07}; as mentioned above, the
seemingly stronger assumption of NFLVR made in \cite{LarZit07} is not necessary
and its results hold under the weaker condition NUBPR. 

Having proven \eqref{equ:main1-v}, we turn to \eqref{equ:main1-u}. 
Thanks to \eqref{equ:simpl}, the conjugacy relationship
  \eqref{equ:conj} takes the following, simple, form in our setting:
  \begin{equation}
  \label{equ:conj2}
  \begin{split}
  p \ue = (q\ve)^{1-p}.
  \end{split}
  \end{equation}
  Therefore, $\ue$ is right differentiable at $\eps=0$, and we have
   \begin{equation*}
   \label{equ:6630}
   \begin{split}
      p \derep{\ue} 
     &= (1-p) (q\vz)^{-p}\, q^2 \vz \Dz = p^2 \uz \Dz. \qedhere
   \end{split}
   \end{equation*}
\end{proof}
\subsection{Remaining proofs}
\begin{proposition} Suppose that $\eta\in\el^{2(1-p)}$ and $\Ld, \Ld\eta\in \el^{1-p}$. Then
\label{pro:v-up}
for all $\eps\geq 0$ we have
 \begin{equation}
 \label{equ:ve-up}
 \begin{split}
  \ve - \vz - \eps q \vz \Dz  \leq   \tot C_v \eps^2 + \tot C'_v \eps^3,
 \end{split}
 \end{equation} 
 where $C_v= |q|\|\eta\|_{\el^{2(1-p)}}^{1/2} + \|\Ld\|_{\el^{1-p}}$ and
 $C'_v =|q|\|\eta \Ld\|_{\el^{1-p}}$. 
\end{proposition}
\begin{proof}
The upper estimate \eqref{equ:upper} and the representation
\eqref{equ:primal1} imply that
\[ \ee{ V(\Ze_T \hHe_T) - V(\Zz_T \hHz_T) - \eps q V(\Zz_T \hHz_T)\eta}
 \leq \EE[A_{\eps}] + \EE[B_{\eps}],\]
 where $A_{\eps}$ and $B_{\eps}$ are defined by \eqref{equ:A-B}, with
 $\Ke=\Kz= \hHz_T$.  As in \eqref{equ:Holder}, we have
 \[  \EE[ B_{\eps}]\leq \tot \eps^2 \|\Ld\|_{\el^{1-p}}.\] To deal with $A_{\eps}$ we note that its structure allows us to
 apply the representation from \eqref{equ:primal1} once again to see
 \[ \oo{q^2} A_{\eps} = \int_0^{\eps} \int_0^{\delta} V(\Zb_T
 \hHz_T) \eta (\eta+\beta \Ld)\, d\beta\, d\delta. \]
 This, in turn, can be estimated, via  H\" older inequality, as in
 \eqref{equ:Holder}, as follows
 \[ \EE[ A_{\eps}] \leq \tot  |q| \eps^2 \sup_{ \beta\in [0,\eps] }   \EE[ (\eta (\eta+\beta \Ld))^{1-p}]^{1+q}
\leq \tot |q| \eps^2 \Big(
\|\eta^2\|_{\el^{1-p}}+\eps \|\eta \Ld\|_{\el^{1-p}}\Big),
\]
 yielding the bound in \eqref{equ:ve-up}.
 \end{proof}
 Unfortunately, the same idea cannot be applied to obtain a similar lower
 bound. Instead, we turn to the primal problem and establish a lower bound
 for it. 
\begin{proposition} \label{pro:u-down} 
Given $\eps_0>0$, assume that $\Ld\in\el^{1-p}$,
 and $\Phi^2 e^{ \eps_0\abs{p} \Phi^-}\in\lone(\PPz)$, where $\PPz$ is defined by \eqref{equ:Girsanov}. Then, 
\[ \ue - \uz - \eps p \uz \Dz \geq -  C_u(\eps) \eps^2 \efor \eps \in
[0,\eps_{0}],\]
 where $C_u(\eps) := \tot p^2 |\uz| \EE^{\PPz}[ \Phi^2 e^{\eps \abs{p} \Phi^-}]$. 
\end{proposition}
\begin{proof}
For $\tX := \EN( \int_0^{\cdot} \hpiz_t\, d\Re_t)$, we have 
$\tX\in\sXe$ so that, by optimality, 
 \begin{equation}
 \label{equ:first-u}
 \begin{split}
    \ue - \uz - p\eps \EE[  U(\hXz_T) \Phi] &\geq 
    \EE[   U(\tX_T) - U(\hXz_T) -  p \eps U(\hXz_T)\Phi].
 \end{split}
 \end{equation}
 Thanks to the form of $\tX$, the right-hand side of \eqref{equ:first-u}
 above can be written as $\EE[ U(\hXz_T) D_{\eps}]$, where
$D_{\eps} = \exp( p \eps \Phi) -1 - p\eps \Phi = 
\int_0^{\eps} \int_0^{\delta}  p^2 \Phi^2 e^{p\beta \Phi}  \, d\beta\, d\delta$.
Thus, 
\begin{align}\label{lower_bound_est} 
\begin{split}
\EE[ U(\hXz_T) D_{\eps}] &=  p^2 \int_0^{\eps} \int_0^{\delta} \EE[ U(\hXz_T) \Phi^2 e^{ p \beta \Phi}]\, d\beta\, d\delta \\&\geq \tot p^2 \eps^2\EE[ U(\hXz_T) \Phi^2 e^{ \eps \abs{p} \Phi^-}].
\end{split}
\end{align}
Therefore,  $\ue - \uz - \eps p  \EE[ U(\hXz_T) \Phi] \geq -  C_u(\eps) \eps^2$, for $\eps\in [0,\eps_0]$ with $C_u$ as in the statement.
 
 It remains to show that 
 $\EE[ U(\hXz_T) \Phi] =  \EE[ U(\hXz_T) \eta]$ which is equivalent to showing $\EE^{\RRz}[ \Phi] =  \EE^{\RRz}[\eta]$  by the definition of $\RRz$.
 We define the local $\RRz$-martingale $\tMp$ by \eqref{equ:tMp}. Therefore, 
 $N = \int_0^{\cdot} \ld'_t\, d\tMp_t$ is also a local martingale.
The  desired
equality
 is therefore equivalent
 to the equality $\EE^{\PPz}[N_T]=0$ by the definition of $\eta$ and $\Phi$. In turn, it is sufficient to
 show that $N$ is an $\sH^2$-martingale under $\PPz$. Since $\ab{N}_T =
 \int_0^{T} (\ld'_t)^2\,  d\ab{M}_t = \Ld$ , H\" older's inequality implies
 that
  \begin{equation*}
  \label{equ:49B0}
  \begin{split}
      \EE^{\RRz}[\ab{N}_T] & = (q\vz)^{-1} \EE[(\hYz_T)^{-q} \Ld]
    \leq (q\vz)^{-1} \EE[\Ld^{1-p}]^{1+q}<\infty.  \qedhere
  \end{split}
  \end{equation*}
\end{proof}
\revised{
\begin{remark}
If one is interested in an error estimate which does not feature the optimal portfolio $\hpiz$ (through $\Phi$), one can adopt an alternative approach in the proof (and the statement) of Proposition \ref{pro:u-down}. More specifically, 
by using $\tX = \hXz\EN( \int_0^{\cdot} \eps \ld' \, d\Re_t)$ as a test process (instead of $\EN( \int_0^{\cdot} \hpiz_t\, d\Re_t)$), one obtains a constant $C_u(\eps)$ which depends only on the primal and dual optimizers $\hXz$ and $\hYz$, in addition to $\ld'$, $\eta$ and $\Lambda$.
\end{remark}
}
\begin{proof}[Proof of Theorem \ref{thm:main2}]
Two of the four inequalities in Theorem \ref{thm:main2} have been
established in Propositions \ref{pro:v-up} and \ref{pro:u-down}. For the
remaining two we use the special form \eqref{equ:conj2}
of the conjugacy relationship between $\ue$ and
$\ve$. Thanks to Proposition \ref{pro:u-down} 
and the positivity of $p\ue$, $q\ve$ and $1+q$, we have
 \begin{equation*}
 \label{equ:4797}
 \begin{split}
    q\Big( \ve-
   \vz - \eps q \vz \Dz\Big) &= (p\ue)^{1+q} - (p\uz)^{1+q} -
    \eps q (p\uz)^{1+q}\Dz.
 \end{split}
 \end{equation*}
 The right-hand side above is further bounded from above,
 for $\eps$ in a (right) neighborhood of $0$, by
\[ F(\eps) := (p\uz + \eps p\uz \Dz - p C\eps^2 )^{1+q} - (p\uz)^{1+q} - 
\eps q  (p\uz)^{1+q} \Dz,\]
where $C$ is the constant from Proposition \ref{pro:u-down}. 
$F$ is a $C^2$-function in some neighborhood of $0$ with
$F(0)=F'(0)=0$; hence, on each compact subset of that neighborhood it is bounded by a constant multiple of $\eps^2$. In particular, we have
\[ \ve - \vz - \eps q \vz \Dz \geq -  C \eps^2,\]
for some $C>0$ and $\eps$ in some (right) neighborhood of $0$. A similar argument, but based on
Proposition \ref{pro:v-up}, shows that \eqref{equ:main2-u} holds, as well. 
\end{proof}

 \proof[Proof of Theorem \ref{thm:main3}] The first part of \eqref{2nd_int} means that $\Lambda \in  \el^{1-p}(\PP)$; hence, the second half of the proof of Proposition \ref{pro:u-down} shows that 
$\EE^{\PPz}[\Phi] = \Delta^{(0)}$. Therefore, the martingale representation \eqref{Itorep} can be written as
\begin{align}\label{Itorep1}
\Phi = \Delta^{(0)} + \int_0^T \gamma^B_{t}\sigma_tdB^{\PPz}_t  +  \int_0^T \gamma^W_{t} dW^{\PPz}_t.
\end{align}

Because the filtration is generated by the  Brownian motions $(B,W)$ we can find $\hat{\nu}^{(0)} \in \sP_{W}^2$ such that the dual optimizer $\hat{H}^{(0)}$ can be represented as
$$
\hat{H}^{(0)} = \mathcal{E}(-\int \hat{\nu}^{(0)} dW).
$$
Therefore, Girsanov's Theorem ensures that under $\PPz$, the processes
$$
dB^{\PPz} := dB +(\lambda -\hat{\pi}^{(0)})\sigma dt, \quad \eand \quad dW^{\PPz}:= dW +\hat{\nu}^{(0)}dt,
$$
are independent Brownian motions. We start with the primal problem and define $\tilde{\pi} := \hat{\pi}^{(0)} + \eps \delta$  with $\delta:=q\gamma^B +\frac{\lambda'}{1-p}\in \sP_{B}^2$. Then we have
\begin{align*}
(\tilde{X})^{p} &:=\mathcal{E}\big(\int\tilde{\pi}dR^{(\eps)})^p\\
&= \big(\hat{X}^{(0)}\big)^pe^{p\int\big (\eps\hat{\pi}^{(0)} \lambda'+ \eps^2 (\delta \lambda'-\frac12\delta^2 ) \big)\sigma^2dt + p\eps \int\delta\sigma dB^{\PPz} }.
\end{align*}
Consequently, by replacing $e^x$ with its Taylor expansion and using that the involved $\PPz$-expectation is finite (here we use the integrability requirement \ref{exp_reg}), we find a function $C_u(\eps) \in O(\eps^3)$ such that
\begin{align}\label{Cu}
\begin{split}
\EE[U(\tilde{X}_T)] 
&= u^{(0)} \EE^{\PPz}\left[e^{p\int_0^T \big(\eps\hat{\pi}^{(0)} \lambda'+ \eps^2 (\delta \lambda'-\frac12\delta^2)\big)\sigma^2dt +p\eps\int_0^T \delta\sigma dB^{\PPz} }\right]\\
& = u^{(0)}\Big(1+p\eps\Delta^{(0)}+ \frac12 p\eps^2\Big\{p(\Delta^{(0)})^2
+\Delta^{(00)}\Big\}\Big)+ C_u(\eps).
\end{split}
\end{align}

We then turn to the dual problem. For the perturbed dual control $\tilde{\nu}:= \hat{\nu}^{(0)} -\eps p\gamma^W\in\sP_{W}^2$ we have
\begin{align*}
&\Big(Z^{(\eps)} \mathcal{E}(-\int \tilde{\nu} dW)\Big)^{-q}\\&= e^{q\int ( \lambda+\eps\lambda' )\sigma dB+q\int ( \hat{\nu}^{(0)}-\eps p\gamma^W )dW+q\frac12\int \big(( \lambda+\eps \lambda' )\sigma ^2 + | \hat{\nu}^{(0)}-\eps q\gamma^W |^2 \big)dt}\\
&= (Z^{(0)}\hat{H}^{(0)})^{-q}e^{\eps q\int  \lambda'\sigma dB^{\PPz}-\eps qp\int \gamma^W dW^{\PPz}+q\frac12\int \big( \eps^2 (\lambda')^2\sigma^2 +\eps^2 p^2|\gamma^W|^2+2\eps\lambda'\pi^{(0)}\sigma^2\big)dt}.
\end{align*}
Since $\tilde{\nu}$ is admissible in the $\eps$-problem we find 
\begin{align*}
&v^{(\eps)} \le \frac1q\EE\left[\Big(Z_T^{(\eps)} \mathcal{E}(-\int_0^T \tilde{\nu} dW)\Big)^{-q}\right] \\&= v^{(0)} \EE^{\PPz}\left[e^{\eps q\int_0^T \lambda'\sigma dB^{\PPz}-\eps qp\int_0^T \gamma^W dW^{\PPz}+q\frac12\int_0^T \big( \eps^2 (\lambda')^2\sigma^2 +\eps^2p^2 |\gamma^W|^2+2\eps\lambda'\pi^{(0)}\sigma^2\big)dt}\right].
\end{align*}
Finiteness of $v^{(\eps)}$ ensures that the $\PPz$-expectation appearing on the last line 
is also finite (recall that $q<0$). As in the primal problem, this allows us to replace $e^x$ with its Taylor series and in turn implies that we can find a function $C_v(\eps) \in O(\eps^3)$ such that
\begin{align*}
&v^{(0)} \EE^{\PPz}\left[e^{\eps q\int_0^T \lambda'\sigma dB^{\PPz}-\eps qp\int_0^T \gamma^W dW^{\PPz}+q\frac12\int_0^T \big( \eps^2 (\lambda')^2\sigma^2 +\eps^2p^2 |\gamma^W|^2+2\eps\lambda'\pi^{(0)}\sigma^2\big)dt}\right]\\&=
v^{(0)}\Big(1+q\eps\Delta^{(0)} + \frac12q\eps^2\Big\{q(\Delta^{(0)})^2 +
\Delta^{(00)}\Big\}\Big)+ C_v(\eps).
\end{align*}
By combining this estimate and \eqref{Cu} with the primal-dual relation \eqref{equ:conj2} we find
\begin{align}
\label{fenchel}
\begin{split}
& u^{(0)}\Big(1+p\eps\Delta^{(0)}+ \frac12 p\eps^2\Big\{p(\Delta^{(0)})^2
+\Delta^{(00)}\Big\}\Big)+ C_u(\eps) \\&\le u^{(\epsilon)} \\&=\frac1p (q v^{(\epsilon)})^{1-p}\\& \le \frac1p\Big(q v^{(0)}\Big[1+q\epsilon \Delta^{(0)} + \frac12q\epsilon^2\Big\{q(\Delta^{(0)})^2 +
\Delta^{(00)}\Big\}+C_v(\epsilon)\Big]\Big)^{1-p}.
\end{split}
\end{align}
The function $x\to x^{1-p}$ is real analytic on $(0,\infty)$. Therefore, the fact that  $C_v\in O(\eps^3)$ ensures that the last line of \eqref{fenchel} agrees with the first line of \eqref{fenchel} up to $O(\eps^3)$-terms. This establishes \eqref{equ:main3-u}. 
A similar argument produces \eqref{equ:main3-v}.
\endproof

\section{Examples}
\label{sec:examples}
\subsection{First examples} We start this section with a short list of
trivial and extreme cases. They are not here to illustrate the power of our
main results, but simply to help the reader understand them better. They
also tell a similar, qualitative, story: loosely speaking, the improvement
in the utility (on the log scale) is proportional both to the base market
price of risk process and to the size of the deviation. Locally, around $\ld$, the 
value function of the utility maximization problem - parametrized by the
market price of risk process $\tilde{\ld}$ - is well approximated by
an exponential function of the form
 \begin{equation}
 \label{equ:exp-approx}
 \begin{split}
    u(\tilde{\ld}) \approx u(\ld) e^{\scl{\tilde{\ld}-\ld}{\hpiz}_{0}},\ewhere
   \scl{\rho}{\pi}_{0} =
   \EE^{\PPz}[ \int_0^T \rho_t\pi_t\, dt],
 \end{split}
 \end{equation}
 where $u(\tilde{\ld})$ and $u(\ld)$ denote the values of the utility-maximization problems with market price of risk processes $\tilde{\ld}$ and $\ld$, respectively.
\begin{example}[Small market price of risk]
\label{exa:small}
Suppose that $\ld\equiv 0$ so that we can think of $\Se$ as the stock price
in a market with a ``small'' market price of risk. Since $\Zz\equiv 1$, it
is clearly the dual optimizer at $\eps=0$ and we have $\hpiz\equiv 0$.
Consequently, under the assumptions of Theorem \ref{thm:main2}, we have
$\PPz=\PP$ and 
\[ \Dz = \EE^{\PPz}[ \int_0^T \ld'_t\, dM_t]=0.\]
It follows that
\[  \ue = \uz + O(\eps^2)\eand \ve=\vz + O(\eps^2),\]
and the effects of $\eps \ld'$ are felt only in the second order,
regardless of the risk-aversion coefficient $p<0$. 
\end{example}
\begin{example}[Deviations from the Black-Scholes model]
\label{exa:Black-Scholes}
Suppose that $M=B$ is an $\FFF$-Brownian 
motion and that $\ld\ne 0$ is a constant process (we also use $\ld$ for the
value of the constant). 
In that case, it is classical that the dual minimizer in the base market
is $\Zz = \EN( - \ld B)$ and, consequently, that 
$\rn{\PPz}{\PP} = \EN( q\ld B)$. It follows that
\[ \Dz = \tfrac{\ld}{1-p} \EE^{\PPz}[ \int_0^T \ld'_t\, dt].\]
As we will see below, this form is especially convenient for computations.
\end{example}
\begin{example}[Uniform deviations] 
\label{exa:constant}
Another special case where it is
particularly easy to compute the (logarithmic derivative) $\Dz$ is when the
perturbation $\ld'$ is a constant process (whose value is also
denoted by $\ld'$). Indeed, in that case
 \begin{equation}
 \label{equ:Dz-unif}
 \begin{split}
   \Dz = \ld' \EE^{\PPz}[ \int_0^T \hpiz_t\, dt].
 \end{split}
 \end{equation}
It is especially instructive to consider the case where the base model is
Black and Scholes' model since everything becomes explicit:
the optimal portfolio is
given by the Merton proportion $\hpiz_t = \ld/(1-p)$, and the
the values $\ue$ and $\ve$ are given by
\[ p\uz = \exp(\tot q \ld^2 T) \eand q\vz = \exp(\tot \tfrac{q}{1-p} \ld^2
T). \]
Using \eqref{equ:Dz-unif} or by performing a 
straightforward direct computation, we easily get
\[ p\Dz = q \ld' \ld T, \]
making the approximation in \eqref{equ:exp-approx} exact.
\end{example}
\subsection{The Kim-Omberg model}
The Kim-Omberg model (see \cite{KimOmb96}) 
is one of the most widely used models for the
market price of risk process. Because the Kim-Omberg model allows for explicit expressions for all quantities involved in CRRA utility maximization it serves as an excellent test
case for the practical implementation of our main results. 

We assume that $\FFF$ is the augmentation of the filtration generated by two independent one dimensional Brownian motions $B$ and $W$ and define $\lambda^\text{KO}$ be the Ornstein-Uhlenbeck
process
\begin{align}
d\lambda_t^\text{KO} &:= \kappa (\theta - \lambda_t^\text{KO})dt + \beta dB_t + \gamma dW_t, \quad \lambda^\text{KO}_0\in\R\label{OU},
\end{align}
where $\kappa, \theta, \beta$ and $\gamma$ are constants. We define the volatility $M_t := B_t$ in what follows.

The following result summarizes the main properties in \cite{KimOmb96}:
\begin{theorem}[Kim and Omberg 1996]\label{thm:KO} Let the market price of risk process be defined by \eqref{OU}, $M:=B$, and let $p<0$. Then there 
exist continuously differentiable functions $a,b,c:[0,\infty)\to \R$ such that for $t\in
[0,T)$ we have
\begin{align*}
-a'(t) & = \alpha_1\, b(t) 
+ \tot \alpha_3\, c(t) - 
\tot \alpha_2\, b^2(t), &
a(T)&=0, \\
-b'(t) &= \alpha_4\, b(t)+ \alpha_1\, c(t)  - \alpha_2\, b(t) c(t),& b(T)&=0, \\
-c'(t) &= -q + 2 \alpha_4\, c(t) - \alpha_2\, c^2(t), & c(T)&=0,
\end{align*}
where $\alpha_1:=\theta\kappa$, $\alpha_2 := (1+q) \beta^2+ \gamma^2$,
$\alpha_3 := \beta^2+\gamma^2$ and $\alpha_4 := q \beta - \kappa$. 
Furthermore, the primal  value function reads
\begin{align}
\label{equ:exact}
u^\text{KO}(x) = \frac{x^p}{p}e^{-a(0) - b(0)\lambda^\text{KO}_0 - \tfrac12 c(T)(\lambda^\text{KO}_0)^2},\quad x>0,
\end{align}
and the corresponding primal optimizer is given by
\begin{align}\label{KO_primal}
\hat{\pi}^\text{KO}_t = \frac{b(t)\beta + \big(c(t)\beta-1\big) \lambda^\text{KO}_t}{p-1},\quad t\in[0,T].
\end{align}
\end{theorem}
For $p<0$, the above Riccati equation describing $c$ has the ``normal
non-exploding solution'' as defined in the appendix of \cite{KimOmb96}.
Therefore, all three functions $a,b$, and $c$ are bounded on any finite
time-interval $[0,T]$ of $(0,\infty)$.

To illustrate our approximation we think of the Kim-Omberg model as a
perturbation of a base model. As base model we will consider the following model with ``totally-unhedgable-coefficients'' (see Example 7.4, p.~305, in \cite{KarShr98}):
\begin{align}\label{KO1}
d\ld_t := \kappa (\theta-\ld_t) dt + \gamma \, dW_t,\quad 
\ld_0:=\lambda^\text{KO}_0. 
\end{align}
This way, $\ldk = \ld+\eps \ld'$, where $\eps = \beta$ and
\begin{align}\label{KO2}
d\ld'_t := - \kappa \ld'_t\, dt + dB_t,\quad \ld'_0:=0.
\end{align}
The following result provides closed-form expressions for our correction terms:
\begin{lemma}\label{lem:KO} Let $(\lambda,\lambda')$ be defined by \eqref{KO1}-\eqref{KO2} and let $p<0$. For the $\eps =0$ model the primal and dual optimizers are given by
\begin{align}\label{KO:base}
\hat{\pi}^{(0)}_t = \frac{\lambda_t}{1-p},\quad \hat{\nu}^{(0)}_t = \gamma \Big(b(t)+c(t)\lambda_t\Big),\quad t\in[0,T].
\end{align}
Furthermore, the processes $(\gamma^B,\gamma^W)$ appearing in the
martingale representation \eqref{Itorep} of $\Phi$  are given by
\begin{align}
\gamma_t^B &=\tfrac1{p-1}(C_2(t)+C_6(t)\lambda_t),\label{KOgammaB}\\
\gamma^W_t &= \tfrac{\gamma}{p-1}(C_4(t)+2C_5(t)
\lambda_t+C_6(t)\lambda_t'),\label{KOgammaW}
\end{align}
where the functions $C_1,C_2,C_4, C_5$ and $C_6$ in 
\eqref{KOgammaB}-\eqref{KOgammaW} satisfy the ODEs
\begin{align*}
-C_1'(t)&= \tilde{b}(t) \, C_4(t)+  \gamma^2 \, C_5(t),& C_1(T)&=0,\\
-C_2'(t)&= \tilde{b}(t) \, C_6(t) - \kappa \, C_2(t),& C_2(T)&=0,\\
-C_4'(t) &= q\, C_2(t) - \tilde{c}(t) \, C_4(t) + 2\tilde{b}(t) \, C_5(t),
& C_4(T)&=0,\\
-C_5'(t) &= q\, C_6(t) -2 \tilde{c}(t) \, C_5(t),& C_5(T)&=0,\\
-C'_6(t) &= - (\kappa +\tilde{c}(t)) \, C_6(t) -1,& C_6(T)&=0,
\end{align*}
on $[0,T)$, with
$(a,b,c)$ as in Theorem \ref{thm:KO} (with $\beta:=0$),
$\tilde{b}(t):=\kappa\theta-\gamma^2 b(t)$ and
$\tilde{c}(t) := \kappa+\gamma^2 c(t) $.
Furthermore, for the measure $\RRz$ defined by \eqref{equ:Girsanov} and for all $T>0$ we have
\begin{align}\label{Delta0_KO}
\Delta^{(0)} &:= \EE^{\RRz}\left[ \int_0^T \lambda'_s\hat{\pi}_s^{(0)} ds\right]= -\frac1{1-p}\Big(C_1(T) + C_4(T) \lambda_0 +C_5(T) \lambda_0^2\Big),
\end{align}
\end{lemma}

\proof The first part follows from Theorem \ref{thm:KO} applied to the case
$\beta:=0$. To find the  martingale representation \eqref{Itorep} we define the function 
$$
f(t,x,\lambda):= \frac{x^p}{p}e^{-a(t) - b(t)\lambda - \tfrac12 c(t)\lambda^2},\quad t\in[0,T],\quad x>0,\quad \lambda\in\R,
$$
where the functions $(a,b,c)$ are as in Theorem \ref{thm:KO}. The martingale properties of
$f(t,\hat{X}^{(0)}_t,\lambda_t)$ and $\hat{X}^{(0)}_t\hat{Y}^{(0)}_t$ as well as the proportionality property $(\hat{X}^{(0)}_T)^{p}\propto\hat{X}^{(0)}_T \hat{Y}^{(0)}_T$ produce 
$$
pf(t, \hat{X}^{(0)}_t,\lambda_t)=p\EE[f(T, \hat{X}^{(0)}_T,\lambda_T)|\sF_t]\propto\EE[ \hat{Y}^{(0)}_T\hat{X}^{(0)}_T|\sF_t] = \hat{X}^{(0)}_t\hat{Y}^{(0)}_t.
$$
By computing the dynamics of the left-hand-side we see from Girsanov's Theorem that the two processes
\begin{align*}
dB^{\RRz}_t &:= 
- q\lambda_tdt + dB_t,\\
\quad dW^{\RRz}_t &:= 
\Big(b(t)+c(t)\lambda_t\Big)\gamma dt + dW_t,
\end{align*}
are independent Brownian motions under $\RRz$. These dynamics and It\^o's Lemma ensure that
\begin{align*}
N_t := \int_0^t &\lambda'_s\lambda_s ds -C_1(t) -C_2(t)\lambda'_t
- C_4(t) \lambda_t -C_5(t) \lambda_t^2-C_6(t)\lambda_t\lambda'_t,
\end{align*}
is a $\RRz$-local martingale.

Because the processes $(\lambda,\lambda')$ remain Ornstein-Uhlenbeck
processes under $\RRz$ and the functions $C_1$-$C_6$ are bounded, $N$ is indeed a
$\RRz$-martingale. Furthermore, thanks to the zero terminal conditions
imposed on $C_1$-$C_6$, we see that
\begin{align*}
\Phi=\tfrac1{1-p}\int_0^T \ld_t \ld'_tdt &= \tfrac1{1-p}N_T = \tfrac1{1-p}N_0 +\int_0^T \gamma^B_t dB^{\PPz}_t +\int_0^T \gamma^W_t dW^{\PPz}_t,
\end{align*}
for $(\gamma^B,\gamma^W)$ defined by \eqref{KOgammaB}-\eqref{KOgammaW}. 
\endproof

\subsubsection{Exact computations.} The proof of Lemma \ref{lem:KO}  shows that 
$$
\Delta^{(0)} :=\EE^{\RRz}\left[ \int_0^T \lambda'_s\hat{\pi}_s^{(0)}
ds\right]= \frac1{p-1}\Big(C_1(0) + C_4(0) \lambda_0 +C_5(0)
\lambda_0^2\Big).
$$
This relation, a similar one (whose exact form and the derivation we omit) 
for the second-order term $\Delta^{(00)}$ of \eqref{Delta00}, and the availability of the exact expression \eqref{equ:exact} for the value function $u^{KO}$ 
allow for an efficient numerical computation of the zeroth-, first-, and second-order
approximation, and their comparison with the exact values.
The model parameters used in the below Table 1 are the calibrated model parameters for the market portfolio reported in Section 4.2 in \cite{LarMun2012} (we ignore
the constant interest rate and constant volatility used in Section 4 in
\cite{LarMun2012}). Moreover, we use negative values of $\eps$ because
the empirical covariation between excess return and the
stock's return is typically negative (see, e.g., the discussion in Section 4.2 in
\cite{LarMun2012}).

Instead of hard-to-interpret expected utility values, we report their certainty
equivalents (i.e., their compositions with the
function $CE:=U^{-1}$; see Remark \ref{rem:761C}(1)).
We set $\delta^{(0)} := p u^{(0)}  \Delta^{(0)}$ and
$\delta^{(00)} :=p u^{(0)} \big( \Delta^{(00)} + p (\Delta^{(0)})^2\big)$. 

\begin{center}
\begin{tabular}{cc||ccc|c}
 $\eps$ & 
 $\lambda_0$& 
 $CE(u^{(0)}) $&
 $CE(u^{(0)}+ \eps \delta^{(0)})$&
 $CE(u^{(0)}+ \eps \delta^{(0)} + \tfrac{\eps^2}{2} \delta^{(00)})$&
 $CE(u^{(\eps)})$\\
\hline\hline
-0.01 & 0.1 & 1.046 & 1.047 & 1.048 & 1.048 \\
- 0.05 & 0.1 & 1.046 & 1.054 & 1.081 & 1.084 \\
- 0.10 & 0.1 & 1.046 & 1.063 & 1.181 & 1.206 \\\hline
- 0.01 & 0.5 & 1.614 & 1.647 & 1.648 & 1.649 \\
- 0.05 & 0.5 & 1.614 & 1.794 & 1.850 & 1.846 \\
- 0.10 & 0.5 & 1.614 & 2.020 & 2.339 & 2.272 \\
\hline
\end{tabular}
\footnotesize 
\begin{quotation}
 \textbf{Table 1.} Certainty equivalents for the zeroth-, first-, and second-order approximations and the exact values in the Kim-Omberg model with
$\beta:=\eps$ and unit initial wealth. The model parameters used are 
$\gamma :=0.04395,\,\kappa := 0.0404,\,\theta := 0.117,\,p:=-1$, and $T:=10$.
\end{quotation}
\end{center}


\subsubsection{Monte-Carlo-based computations}
One of the advantages of our approach is that it lends
itself easily to computational methods based on Monte-Carlo (MC) simulation. For the Kim-Omberg model we use the standard explicit Euler scheme from MC simulation to compute the involved quantities of interest. In other words, we do not rely on the availability of exact expressions for the value
functions or the correction terms $\Delta^{(0)}$ and $\Delta^{(00)}$.

For a portfolio $\pi$ and the
model-perturbation parameter $\eps$, the  constant
$\text{CE}^{(\eps)}(\pi)\in(0,\infty)$ is uniquely defined by
\begin{align}\label{approxCE}
U\left(\text{CE}^{(\eps)}(\pi)\right) = \EE\left[U\left(\EN\big(\int_0^T
\pi_t\, d\Re_t\big)\right)\right].
\end{align}
In other words, $\text{CE}^{(\eps)}{(\pi)}$ is the dollar amount whose
utility value matches that of the expected utility an investor would obtain
in the $\eps$-model who uses the strategy $\pi$. We remind the reader that
$\hat{\pi}^{(0)}$ denotes the optimizer in the base ($\eps=0$)
model, $\tilde{\pi}^{(\eps)}$ is the second-order improvement (as in
\ref{optimal_control_correction} above)
of $\hat{\pi}^{(0)}$,
and $\hat{\pi}^{(\eps)}$ is the exact optimizer in the $\eps$-model.
Both quantities 
$\text{CE}^{(\eps)}(\hat{\pi}^{(0)})$ and 
$\text{CE}^{(\eps)}(\tilde{\pi}^{(\eps)})$ serve as lower bounds for the exact
value $\text{CE}(u^{(\eps)})$. The second one, which we also denote by
\begin{align}\label{def_LB}
 \text{LB} := \text{CE}^{(\eps)}(\tilde{\pi}^{(\eps)}),
 \end{align}
is second-order optimal and appears in our simulations. 
To obtain a corresponding upper bound, we simulate the dynamics of the dual
process, based on $\eqref{equ:conj2}$ and the
second-order optimal dual control $\tilde{\nu}$ defined by
\eqref{optimal_dual_control_correction}. We define
\begin{align}\label{def_UB}
\text{UB} := U^{-1} \left( \frac1p\EE\left[\Big(Z_T^{(\eps)}
\mathcal{E}(-\int_0^T \tilde{\nu}_u dW_u)\Big)^{-q}\right]^{1-p} \right).
\end{align}

To quantify the simulation errors, we report the $95\%$-confidence intervals based on MC simulated values of  $\text{CE}^{(\eps)}(\hat{\pi}^{(0)})$, LB, and UB in the below Table 2. The value $\text{CE}(u^{(\eps)})$,
computed without MC simulation and included for comparison only, 
is exact to $3$ decimal places. 
\begin{center}
\begin{tabular}{cc||ccc|c}
 $\eps$ & 
 $\lambda_0$& 
 $\text{CE}^{(\eps)}(\hat{\pi}^{(0)}) $ &
 $\text{LB}$ &
 $\text{UB}$ &
 $\text{CE}(u^{(\eps)})$\\
\hline\hline
-0.01& 	0.10& $[1.047, 1.048]$& [1.048, 1.049]& [1.048, 1.049]& 1.048\\  
-0.05& 	0.10& $[1.052, 1.053]$& [1.083, 1.084]& [1.083, 1.085]& 1.084\\  
-0.10& 	0.10& $[1.057, 1.058]$& [1.200, 1.201]& [1.204, 1.208]& 1.206\\  
\hline
-0.01& 	0.50& $[1.644, 1.649]$& [1.647, 1.653]& [1.646, 1.657]& 1.649\\  
-0.05& 	0.50& $[1.760, 1.764]$& [1.844, 1.850]& [1.843, 1.857]& 1.846\\ 
-0.10& 	0.50& $[1.868, 1.871]$& [2.248, 2.256]& [2.266, 2.286]& 2.272\\  
\hline
\end{tabular}
\begin{quotation}
\footnotesize \textbf{Table 2.} 
$95\%$-confidence intervals for certainty equivalents for the upper and lower bounds as well as the base model optimizer $\hat{\pi}^{(0)}$ for the Kim-Omberg model. The true exact values for the $\varepsilon$-model are included in the last column for comparison. Except for the last column, the numbers are based on MC simulation using Euler's scheme with one million paths each with time-step size $0.001$. The model parameters are the same as in Table 1.
\end{quotation}
\end{center}

In Table 2 we note the significant difference 
between the performance of the base-model optimizer
$\hat{\pi}^{(0)}$ and its second-order improvement
$\tilde{\pi}^{(\eps)}$; especially for larger values of $\eps$. Furthermore, the lower and upper bounds appear to be quite tight.

\subsection{Extended affine models} 
\label{sse:extaff}
We turn to a class of models for which no closed-form expressions for the
value functions $u$ and $v$ seem to be available. It constitutes the
main example of the class of so-called extended-affine specifications of the
market price of risk models introduced by \cite{CheFilKim07}. 

As in the
Kim-Omberg model above we let the augmented filtration be generated by
two independent Brownian motions $B$ and $W$. The
central role is played by the following Feller process $F$
\begin{align}
dF_t&:= \kappa (\theta - F_t)dt +  \sqrt{F_t}\, \big( \beta dB_t + \gamma dW_t\big), \quad F_0>0,\label{Feller}
\end{align}
where $\kappa, \theta, \beta$ and $\gamma$ are strictly 
positive constants such that the (strict) Feller condition
$2\kappa\theta > \beta^2 +\gamma^2$ holds. This ensures, in particular,
that $F$ is strictly positive on $[0,T]$, almost surely. 
Unlike in the Kim-Omberg model, the appropriate volatility normalization turns out to be 
$\sqrt{F_t}$; that is, we define 
 \begin{equation}
 \label{equ:M-F}
 \begin{split}
   M := \int_0^{\cdot} \sqrt{F_t}\, dB_t.
 \end{split}
 \end{equation}
A particular extended affine specification of the market price of risk process considered in \cite{CheFilKim07}  is given by
\begin{align}\label{CFK}
\lambda^\text{CFK}_t := \frac{\eps}{F_t}+1,
\end{align}
where $\eps$ is a (positive or negative) constant. Unless $\eps=0$,
there is 
currently no known closed-form 
solution to the corresponding optimal investment problem (Theorem 4.5 in \cite{GR15} expresses the corresponding value function as an infinite sum of weighted generalized Laguerre polynomials). However, for
$\eps =0$, the resulting model is covered by the analysis in
\cite{Kra05}. Therefore, we choose the constant 
market price of risk process
\[ \lambda_t:= 1\] for the base model whereas we define the 
perturbation process $\lambda'$ by
\begin{align}\label{CV}
\lambda'_t := \frac1{F_t}.
\end{align}
\begin{theorem}[Kraft 2005]\label{thm:CV} For $p<0$ there exist continuously differentiable functions $a,b:[0,T)\to \R$ such
that 
\begin{align*}
- a'(t)&= \alpha_1\, b(t),& a(T)&=0, \\
-b'(t) &= \alpha_4 \, b(t) - \tot \alpha_2
\, b^2(t) - \tot q,& b(T)&=0,
\end{align*}
where $\alpha_1:=\theta\kappa$, $\alpha_2 := (1+q) \beta^2+ \gamma^2$,
 and $\alpha_4 := q \beta - \kappa$.  The value function of the utility-maximization problem with $\ld:= 1$ and $M$ as in \eqref{equ:M-F} is given by 
$$
u^{(0)}(x) = \frac{x^p}{p}e^{-a(0) - b(0) F_0},\quad x>0.
$$
The corresponding primal and dual optimizers are given by
\begin{align}\label{CV_primal}
\hat{\pi}^{(0)}_t = \frac{b(t)\beta-1}{p-1},\quad \hat{\nu}^{(0)}_t = b(t)\gamma\sqrt{F_t},\quad  t\in[0,T].
\end{align}
\end{theorem}

To check the conditions of our main
theorems, we use the explicit expression in \cite{HurKuz2008}, Theorem
3.1, for the Laplace transform 
\[ L(a_1, a_2) := \EE[ \exp( a_1 Q + a_2 \Ld)],\quad Q:=\int_0^T F_s\, ds,\quad \Ld:=\int_0^T \tfrac{1}{F_s}\, ds.\]
It is shown in \cite{HurKuz2008} that $L$ is finite in some neighborhood
of $0$ under the strict Feller condition $2\kappa\theta > \beta^2
+\gamma^2$. This implies that both $\Ld$ and $Q$ have a finite exponential
moment. In particular, H\"older's inequality with exponents $-1/q$ and $(1-p)$ implies that
$$
\EE^{\PPz}[\Lambda]=\frac1{qv^{(0)}}\EE[(\hat{Y}_T^{(0)})^{-q}\Lambda]\le \frac1{qv^{(0)}}\EE[\Lambda^{1-p}]^{\frac{1}{1-p}}<\infty.
$$ Thanks to the deterministic behavior of $\hat{\pi}^{(0)}$ in
\eqref{CV_primal}, the  martingale representation \eqref{Itorep} of $\Phi$  holds
with $\gamma^B = \gamma^W=0$. Consequently, we have 
$$
\Phi :=  \int_0^T  \hpiz_s\, ds =\Delta^{(0)},\quad 
\Delta^{(00)} :=\tfrac{1 }{1-p} \EE^{\PPz}[\Lambda].
$$
To verify that \eqref{exp_reg} holds, we can use H\"older's inequality (twice) with exponents $-1/q$ and $(1-p)$ to see
\begin{align*}
\EE^{\PPz}\left[e^{-\frac12
\eps^2 \frac{p}{(1-p)^2}\Lambda + q\eps \int_0^T \frac1{\sqrt{F_t}} dB^{\PPz}_t}\right]&\le
\EE^{\PPz}\left[e^{-\frac12
\eps^2 (p+q)\Lambda}\right]^{\frac1{1-p}}\\
&\le
\frac1{qv^{(0)}}\EE\left[e^{-\frac12
\eps^2 (1-p)(p+q)\Lambda}\right]^{\frac1{(1-p)^2}},
\end{align*}
which is finite for $\eps>0$ small enough. This allows Theorem
\ref{thm:main3} to be invoked for $\eps>0$ small enough.
The second-order optimal controls $(\tilde{\pi},\tilde{\nu})$  are then well defined
by \eqref{optimal_control_correction} and
\eqref{optimal_dual_control_correction}, and read
\begin{align}
\tilde{\pi} := \hat{\pi}^{(0)} + \eps \tfrac{\lambda'}{1-p},\quad
\tilde{\nu} := \hat{\nu}^{(0)}.
\end{align}

Table 3 is the analogue of Table 2 for the extended affine model with parameters taken from Figure 4 in Section 3.3 in \cite{LarMun2012}.  The methodology and the simulated quantities are the same as for Table 2.
\newpage
\begin{center}
  \begin{tabular}{cc||ccc}
$\eps$ & $F_0$ & $\text{CE}^{(\eps)}(\hat{\pi}^{(0)})$  & LB &
 UB \\ \hline\hline
$0.10$ & $0.01$ & $[1.724, 1.726]$& $[10.159, 10.399]$&$[10.226, 10.481]$  \\
$0.05$& $0.01$ &  $[1.342, 1.343]$& $[2.141, 2.151]$&$[2.131, 2.149]$  \\
$0.01$& $0.01$& $[1.097, 1.098]$& $[1.118, 1.119]$&$[1.117, 1.120]$  \\
\hline
$0.10$& $0.05$& $[1.728, 1.729]$& $[9.660, 9.877]$&$[9.766, 10.000]$  \\
$0.05$& $0.05$& $[1.344, 1.345]$& $[2.105, 2.115]$&$[2.102, 2.120 ]$  \\
$0.01$& $0.05$& $[1.099, 1.100]$& $[1.119, 1.121]$&$[1.117, 1.121]$  \\
\hline
\end{tabular}
\begin{quotation}
\footnotesize \textbf{Table 3.} $95\%$-confidence intervals for certainty equivalents for the upper and lower bounds as well as the base model optimizer $\hat{\pi}^{(0)}$ for the extended affine model. The parameter values are 
$\kappa := 5$, $\theta := 0.0169$, $\beta := -0.1$, $\gamma :=0.1744$, $p:=-1$,
and $T:=10$. The numbers are based on MC simulation using Euler's scheme with one million paths each with time-step size $0.001$.
\end{quotation}

\end{center}

The zeroth order approximation $\text{CE}^{(0)}(\hat{\pi}^{(0)})$ produces the certainty equivalent values 
$$
\text{CE}^{(0)}(\hat{\pi}^{(0)}) = 1.043 \; (F_0=0.01),\quad\text{and }\quad   \text{CE}^{(0)}(\hat{\pi}^{(0)}) = 1.045 \;(F_0=0.05).
$$ 
Perhaps even more than in the Kim-Omberg model, the numbers in Table 3 above illustrate the superiority of the second-order approximations (columns 4 and 5) over its first-order version (column 3) as well as the zeroth order values reported above. Again, the bounds in Table 3 appear quite tight when compared to the first-order approximations for moderate values of $\varepsilon$. 

\bibliographystyle{amsalpha}
\bibliography{LarMosZit14}
\end{document}